\documentclass[11pt,a4paper,reqno]{amsart}
\usepackage{amsthm,amssymb,amsmath,graphicx,enumerate}
\usepackage{mathrsfs,setspace,pstricks,multicol,latexsym}
\usepackage{diagbox}
\usepackage{epstopdf}
\usepackage{a4wide}
\usepackage{appendix}
\usepackage{hyperref}
\usepackage{xcolor}
\usepackage{subfigure}
\usepackage{graphics}
\usepackage[vlined,ruled,linesnumbered]{algorithm2e}
\usepackage{mathtools}
\newcolumntype{P}[1]{>{\centering\arraybackslash}p{#1}}
\newcommand{\abs}[1]{ \left\lvert#1\right\rvert}
\newcommand{\norm}[1]{\left\lVert#1\right\rVert}

\DeclarePairedDelimiter\floor{\lfloor}{\rfloor}
\newtheorem{theorem}{Theorem}

\newtheorem{proposition}[theorem]{Proposition}

\newtheorem{lemma}[theorem]{Lemma}

\newtheorem{remark}[theorem]{Remark}

\begin{document}
\author{Anindya Goswami}
\address{IISER Pune, India}
\email{anindya@iiserpune.ac.in}
\author{Kuldip Singh Patel*}\thanks{* Corresponding author}
\address{IIT Patna, India}
\email{kspatel@iitp.ac.in}
\title[Matrix method stability]{Matrix method stability and robustness of compact schemes for parabolic PDEs}
\thanks{Authors acknowledge the support from Government of India for the financial support under the grant no. DST/INT/DAAD/P-12/2020 and 02011-32-2023-R$\&$D-II-13347.}
\begin{abstract}
The fully discrete problem for convection-diffusion equation is considered. It comprises compact approximations for spatial discretization, and Crank-Nicolson scheme for temporal discretization. The expressions for the entries of inverse of tridiagonal Toeplitz matrix, and Gerschgorin circle theorem have been applied to locate the eigenvalues of the amplification matrix. An upper bound on the condition number of a relevant matrix is derived. It is shown to be of order $\mathcal{O}\left(\frac{\delta v}{\delta z^2}\right)$, where $\delta v$ and $\delta z$ are time and space step sizes respectively. Some numerical illustrations have been added to complement the theoretical findings. 
\end{abstract}
\maketitle
\section{Introduction}\label{sec:intro}
The convection–diffusion equation is ubiquitous in several phenomena, for example, option pricing problems in stock market \cite{BlaS73}, computational fluid dynamics \cite{roach1976computational}, and in various other physical systems \cite{isenberg1973heat, fattah1985dispersion}. The analytical solution of the convection-diffusion equations is only obtained in a few cases, and it is not available in general. Therefore, a rich theory of numerical methods is essential to solve such problems efficiently and accurately. In literature, various numerical methods, for example finite difference method (FDM), finite element methods, spectral methods, wavelet based method etc., have been developed for solving convection-diffusion equation \cite{AchP05, Can98, MKAW20}. 
\par In fact, high-order accurate FDM can be developed by increasing the number of grid points in a computational stencil. However, the implementation of boundary conditions becomes tedious in those cases. Moreover, the corresponding coefficient matrices in fully discrete problem have more non-zero entries. Therefore, high-order accurate FDMs were developed using compact stencils, which utilizes the same number of grid points and provides better rate of convergence. These are known as a compact schemes and has also been applied to solve convection-diffusion equations \cite{Spotz01, KPMA17}. In these schemes, the non-zero entries of the coefficient matrix are cumbersome but tractable. 
\par The stability analysis of numerical schemes is pervasive in the numerical solution of PDEs, and it has been discussed in great detail by various authors \cite{sousa2001finite, Tref96}. Various approaches have been used in the literature to investigate the stability of finite difference schemes. A few of them are: (i) von Neumann approach, (ii) matrix method, (iii) energy method, (iv) normal mode analysis etc. The von-Neumann approach is suitable for pure initial value problems and problems with periodic boundary conditions. The matrix method, which is applicable for the boundary value problems and also for problems with variable coefficients, involves the estimation of eigenvalues of the amplification matrix. The energy method often leads to the sufficient conditions for stability, however it is often tedious to obtain the bounds in $l^2$ norm. The application of normal mode theory emerges as yet another valuable tool, offering an alternative perspective for assessing the stability of numerical schemes.
As mentioned above, the matrix method seems to apply to the larger class of problems as compared to von-Neumann approach. 

\par In this paper, the stability of the compact schemes for one-dimensional convection-diffusion equations with constant coefficients is studied using the matrix method. The stability of the same scheme has already been proved with less effort by following von Neumann approach in \cite{Spotz95}. 
The literature on matrix method stability analysis for compact schemes is absent even for constant coefficient PDEs. This is because the amplification matrix is not sparse in this case, and the entries of this matrix are intractable. The matrix method analysis involves writing the difference equation in one time-step of the form $U^{n+1} = AU^n$. Then the eigenvalues of the matrix $A$ are estimated, and the region of stability is taken as the region where the spectral radius of $A$ is less than one. However, our objective in this manuscript is twofold. First, the proposed method to estimate the eigenvalues of a complicated matrix is completely novel. Second, the method is applicable to the variety of extensions of the problem taken into consideration in this paper, where other methods have practical limitations. Additionally, the presented analysis provides a novel theoretical approach for estimating the eigenvalues of a complicated matrix, which may have applications in other related fields. 

\par The absence of matrix method stability analysis for compact schemes poses a gap in the literature. To bridge the existing gap, a novel approach is proposed for the matrix method stability analysis utilising the Gerschgorin Circle Theorem (GCT). The crank-Nicolson method is used for temporal semi-discretization, and compact scheme is applied to discretize the space variable, which leads to complicated system of linear equations. The first complexities we encountered is to estimate the eigenvalues of the amplification matrix for the proposed scheme, as it involves matrix inversion. To overcome this, the inversion of relevant Toeplitz matrix \cite{Mallik01} is utilized to locate the eigenvalues. The proposed methodology offers valuable insights into the process of locating these eigenvalues. Another challenge was to comment on the robustness of the proposed scheme, which has direct relation with the condition number of the amplification matrix. The discussion on the condition number of the amplification matrix of compact schemes is absent in the literature to the best of our knowledge. An upper bound on the condition number of the matrix is obtained, which needs to be inverted for computing the amplification matrix. The upper bound is shown to be of order $\mathcal{O}\left(\frac{\delta v}{\delta z^2}\right)$, where $\delta v$ and $\delta z$ are time and space step sizes, respectively. A few numerical experiments are added to validate the assumptions for a wide range of parameter values. Some more numerical experiments are provided to illustrate the theoretical findings. 

\par The present paper is structured as follows: The fully discrete problem for convection-diffusion equation is presented in Sec. \ref{sec:constant_PDE}. The stability of the compact scheme for convection-diffusion equation is proved in Sec. \ref{sec:stability}. Sec. \ref{sec:condition_number} presents the results related to the condition number. Numerical illustrations are given in Sec. \ref{sec:numerical} to support theoretical findings. Sec. \ref{sec:conclu} includes the concluding remarks with some future research directions.

\section{The Fully Discrete Problem} \label{sec:constant_PDE}
\noindent Let $\alpha_1$ and $\alpha_2$ be two constants where $\alpha_2>0$, and $\Omega_x = (x_l, x_r)$ be a finite open interval. Then, a convection-diffusion equation on $\Omega_x$ can be written as follows:
\begin{equation}
\frac{\partial \psi}{\partial v}(v,x)+\alpha_1\frac{\partial \psi}{\partial x}(v,x)-\alpha_2\frac{\partial^2 \psi}{\partial x^2}(v,x)=0,\label{eq:oned_main_equation1}
\end{equation}
where $x \in \Omega_x$, and $0 \leq v \leq T$ for some positive constant $T$. If we take $u = \frac{\alpha_2}{\alpha_1}\psi,$ and $z=\frac{\alpha_1}{\alpha_2}x$ in above equation (\ref{eq:oned_main_equation1}), we have
\begin{equation}
\frac{\partial u}{\partial v}(v,z)+c\frac{\partial u}{\partial z}(v,z)-c\frac{\partial^2 u}{\partial z^2}(v,z)=0,\label{eq:oned_main_equation_2}
\end{equation}
for $z\in \Omega_z=(z_l,z_r)$, and $c = \frac{\alpha_1^2}{\alpha_2}.$ Note that $c$ is positive $\forall$ $\alpha_1$ and positive $\alpha_2$. Moreover, we associate the following initial and boundary conditions with Eq. (\ref{eq:oned_main_equation_2})
\begin{align}
u(0,z)&=f(z), \quad z \in (z_l,z_r)\label{eq:initial_condition_1d}\\
u(v,z_l)&=g_1(v), \quad u(v,z_r)=g_2(v),\quad v> 0, \label{eq:oned_boun}
\end{align}
under the assumptions that $g_1$ and $g_2$ are smooth functions. Further, $g_1(0)=f(z_l)$ and $g_2(0)=f(z_r)$. 
 \par Now, the fully discrete problem for Eq. (\ref{eq:oned_main_equation_2}) is presented. For the sake of simplicity, a uniformly spaced mesh is considered in both temporal and spatial domain. For fixed $N$, we define $z_q =z_l + q\delta z, 0 \leq q \leq N$ for a fixed space step size $\delta z=\frac{z_r-z_l}{N}$. Also for fixed $M$, consider the $m^{th}$ time step as $m\delta v$ for constant time step size $\delta v$ and $m=0,1,...,M$. Let $u^m_q$ denote the solution of (\ref{eq:oned_main_equation_2}) at $m^{th}$ time level and at space grid point $z_q$. Following Eq. (6) in \cite{Spotz01}, the fully discrete problem for (\ref{eq:oned_main_equation_2}) using Crank-Nicolson compact scheme at space grid point $z_q$ and time level $m$ is
\begin{align}\label{eq:oned_fully_discrete_compact}
\Delta_v^+ u_q^m&-\frac{\delta z^2}{12}\left(\Delta_v^+\Delta_z u_q^m-\Delta_v^+\Delta_{zz} u_q^m\right)+\frac{c}{2}\left(\Delta_z u_q^m-\left(1+\frac{\delta z^2}{12}\right)\Delta_{zz} u_q^m\right)\nonumber \\
&+\frac{c}{2}\left(\Delta_z u_q^{m+1}-\left(1+\frac{\delta z^2}{12}\right)\Delta_{zz} u_q^{m+1}\right)=\mathcal{O}(\delta v^2, \delta z^4).
\end{align}
Here $\Delta_v^+u_q^m$, $\Delta_{z}u^m_q$, and $\Delta_{zz}u^m_q$ represent finite difference approximations for first order time derivative, first order space derivative, and second order space derivative of $u$ respectively at $m^{th}$ time level and space grid point $z_q$. The expressions for the same are as follows:
\begin{equation}\label{eq:firstf_2}
 \Delta_v^+u_q^m=\frac{u_{q}^{m+1}-u_{q}^{m}}{\delta v},\:\: \Delta_{z}u^m_q=\frac{u^m_{q+1}-u^m_{q-1}}{2\delta z},\:\: \mbox{and}\:\: \Delta_{zz}u^m_{q}=\frac{u^m_{q+1}-2u^m_{q}+u^m_{q-1}}{\delta z^2}.
\end{equation} 
If $U^m_q$ denotes the approximate value of $u^m_q$, then using relation (\ref{eq:firstf_2}) in (\ref{eq:oned_fully_discrete_compact}) and rearranging the terms, we get the following fully discrete problem for all $1 \leq q \leq N-1$
\begin{align}\label{eq:oned_fully_discrete_compact_1}
&U^{m+1}_{q-1}\left(\frac{2+\delta z}{24\delta v}-\frac{c}{4 \delta z}-\frac{c+\frac{c\delta z^2}{12 }}{2\delta z^2}\right)+U^{m+1}_q\left(\frac{5}{6 \delta v}+\frac{c+\frac{c\delta z^2}{12}}{\delta z^2}\right)+\nonumber\\
&U^{m+1}_{q+1}\left(\frac{2-\delta z}{24\delta v}+\frac{c}{4 \delta z}-\frac{c+\frac{c\delta z^2}{12 }}{2\delta z^2}\right)=U^{m}_{q-1}\left(\frac{2+\delta z}{24\delta v}+\frac{c}{4 \delta z}+\frac{c+\frac{c\delta z^2}{12}}{2\delta z^2}\right)\nonumber\\
&+U^{m}_q\left(\frac{5}{6 \delta v}-\frac{c+\frac{c\delta z^2}{12}}{\delta z^2}\right)+U^{m}_{q+1}\left(\frac{2-\delta z}{24\delta v}-\frac{c}{4 \delta z}+\frac{c+\frac{c\delta z^2}{12}}{2\delta z^2}\right),
\end{align}
with $U^m_0=g_1(v_m)$, and $U^m_N=g_2(v_m)$, for all $m>0$. Now we will prove the stability of the fully discrete problem (\ref{eq:oned_fully_discrete_compact_1}) in the following section.
\section{Stability}\label{sec:stability}
\par In this section, we prove the stability of fully discrete problem (\ref{eq:oned_fully_discrete_compact_1}). Suppose $U^{m}$ denotes the vector  $[U_1^m,U_3^m,\cdots,U_{N-1}^m]^{T}$, where $[\cdot]^{T}$ denotes the transpose of the vector. We also introduce the following constants, depending on $\delta z $ and $\delta v$:
$$c_1= \frac{(2+\delta z)}{24\delta v}, \quad c_2=\frac{5}{6 \delta v}, \quad c_3= \frac{2-\delta z}{24 \delta v},$$ $$y_1=-\frac{c}{4 \delta z}-\frac{c+\frac{c\delta z^2}{12 }}{2\delta z^2},\quad y_2=\frac{c+\frac{c\delta z^2}{12}}{\delta z^2}, \quad \mbox{and}  \quad y_3=\frac{c}{4 \delta z}-\frac{c+\frac{c\delta z^2}{12 }}{2\delta z^2}.$$ 
Set $b:=\frac{\delta v}{\delta z^2}$, to simplify 
\begin{align} \label{y123}
y_1=\frac{-c}{2 \delta v}\left[\left(1+ \frac{\delta z}{2}\right)b+\frac{\delta v}{12}\right], 
\quad y_2=\frac{c}{2 \delta v} \left(2b+\frac{\delta v}{6}\right), \quad \mbox{and} \quad y_3=\frac{-c}{2 \delta v}\left[\left(1-\frac{\delta z}{2}\right)b+\frac{\delta v}{12}\right].
\end{align}
The fully discrete problem (\ref{eq:oned_fully_discrete_compact_1}) can be written as
\begin{equation}\label{eq:xy}
(X+Y)U^{m+1}=(X-Y)U^m+F^m, 
\end{equation}
where 
\begin{equation}\label{eq:matrix}
X\vcentcolon =\left[ 
\begin{array}{c c c c c c c c c} 
  c_2 & c_3 & & & \cdots & & & &\\ 
  c_1 & c_2 & c_3 & & \cdots & & & & \\ 
 & & & & \vdots & & & &\\
 & & & & \cdots & & c_1 & c_2 & c_3 \\ 
 & & & & \cdots & & & c_1 & c_2
\end{array} 
\right], \quad
Y\vcentcolon =\left[ 
\begin{array}{c c c c c c c c c} 
  y_2 & y_3 & & & \cdots & & & &\\ 
 y_1 & y_2 & y_3 & & \cdots & & & & \\ 
 & & & & \vdots & & & &\\
 & & & & \cdots & & y_1 & y_2 & y_3 \\ 
 & & & & \cdots & & & y_1 & y_2
\end{array} 
\right],
\end{equation}
and $$F^m=[(c_1-y_1)g_1(v_m)-(c_1+y_1)g_1(v_{m+1}),0,...,0,(c_3-y_3)g_2(v_m)-(c_3+y_3)g_2(v_{m+1})]^T.$$ 
\noindent We assume that $0< \delta z <2$ now onwards to ensure positivity of the constant $c_3$. Evidently, $X$ is diagonally dominant, and hence invertible. Therefore, using $W:=X^{-1}Y$, Eq. (\ref{eq:xy}) can be rewritten as
\begin{equation}\label{eq:w}
(I+W)U^{m+1}=(I-W)U^{m}+X^{-1}F^m.
\end{equation}
\begin{theorem}\label{theorem:stability}
Assume that the real parts of the eigenvalues of $W$ are positive. Then fully discrete problem (\ref{eq:oned_fully_discrete_compact_1}) is stable.
\end{theorem}
\begin{proof}
Since real part of eigenvalues of $W$ are positive, $(I+W)^{-1}$ exists. Therefore, we can rewrite equation \eqref{eq:w} as $$U^{m+1}=HU^{m}+(X+Y)^{-1}F^m,$$ where $H=(I+W)^{-1}(I-W)$. Note that if $\rho$ and $\beta$ is a pair of eigenvalue and eigenvector of $W$, we get
\begin{align*}
H\beta=\frac{1-\rho}{1+\rho}\beta.
\end{align*}
Consequently, $\frac{1-\rho}{1+\rho}$ is an eigenvalue of $H$, the amplification matrix. Hence the fully discrete problem (\ref{eq:oned_fully_discrete_compact_1}) is stable provided $\abs{\frac{1-\rho}{1+\rho}} < 1$ for each eigenvalue $\rho$ of $W$. Since we have assumed that the real part of $\rho$ is positive, the modulus of $1-\rho$ is smaller than that of $1+\rho$. Thus $$\abs{\frac{1-\rho} {1+\rho}}<1.$$ Hence the fully discrete problem (\ref{eq:oned_fully_discrete_compact_1}]) is stable. 
\end{proof}

\begin{remark} \label{remark:N2}
The positivity of real part of eigenvalue of $W$, the assumption in Theorem \ref{theorem:stability}, needs to be verified. For $N=2$, the matrices $X$, $Y$, and $W$ are just scalars, and their values are
$$
X=c_2, \quad Y=y_2, \quad \mbox{and} \quad  W=\frac{y_2}{c_2}= \frac{3 c}{5} \left(2b+\frac{\delta v}{6}\right)>0
$$
using \eqref{y123}. 
Therefore, the assumption in Theorem \ref{theorem:stability} holds true for $N=2$ case.
That assumption is also shown to be true theoretically $N=3$ in Proposition \ref{lemma:eigenless1} with no additional assumption on the model parameters. Hence, this is not an unrealistic assumption. Moreover in Section \ref{sec:numerical}, this assumption is verified numerically for wide range of $N$.
\end{remark}

\par To verify the assumption for $N=3$ case, the elements' expressions of matrix $X^{-1}$ are required. Since $X$ is Toeplitz, we use the expression that appears on pp. $137$ in \cite{Mallik01}.  For all $1 \leq q, q' \leq N-1$ 
\begin{align}\label{eq:tauinv}
(X^{-1})_{q,q'}=
\frac{(-1)^{q-q'}}{\sqrt{c_1c_3}}\left(\sqrt{\frac{c_1}{c_3}}\right)^{q-q'}\frac{p_{(q\wedge q')-1} \times p_{N-1-(q\vee q')} }{p_{N-1}}\left(\frac{c_2}{2\sqrt{c_1c_3}}\right)
\end{align}
where $q\wedge q'$ and $q\vee q' $ are minimum and maximum of $\{q,q'\}$ respectively and 
\[
p_n(x)=(2x)^n\left[1+\sum_{n'=1}^{\floor{n/2}}(-1)^{n'} \binom{n-n'}{n'}\left(\frac{1}{4x^2}\right)^{n'}\right].
\]
Note that, due to the assumption $\delta z <2$, $c_3$ is positive and hence the expression in \eqref{eq:tauinv} is real. In order to locate the eigenvalues of matrix $W$, the following result is borrowed from pp. $61$ in \cite{Smith78}:\\

\noindent {\bf Gerschgorin Circle Theorem:}
Suppose $M_q$ denotes the sum of the modulus of the elements of $q^{th}$ row of a matrix $B=(B_{q,q'})_{1 \leq q, q' \leq N-1}$ by excluding the diagonal element $B_{q,q}$. Then each eigenvalue of the matrix $B$ lies inside 
$\cup_{1\le q\le N-1} D(B_{q,q}, M_q)$
where $D(a,r)$ denotes the disc with center $a$ and radius $r$ on the complex plane.

\begin{proposition}\label{lemma:eigenless1} Assume $N=3$ and $\delta z<2$. If $\rho$ is an eigenvalue of $W$, then the real part of $\rho$ is positive for sufficiently small $\delta z>0$. \end{proposition}
\begin{proof}
\par Using the fact that $Y$ is tri-diagonal, we can write for $1 \leq q,q' \leq N-1$
\begin{align}\label{eq:wij1}
\nonumber W_{q,q'}=&(X^{-1})_{q,q'-1}Y_{q'-1,q'}+(X^{-1})_{q,q'}Y_{q',q'}+(X^{-1})_{q,q'+1}Y_{q'+1,q'}\\
= &(X^{-1})_{q,q'-1}y_3+(X^{-1})_{q,q'}y_2+(X^{-1})_{q,q'+1}y_1.
\end{align}
\noindent In the above expression we mean $$(X^{-1})_{q,q'}=0, \: \text{if either of} \: q \:\text{or} \: q' \: \text{is not in} \:\{1,\ldots, N-1\}.$$ The fraction  $\frac{c_2}{2\sqrt{c_1c_3}}$ appearing in \eqref{eq:tauinv} has a value $\frac{10}{\sqrt{4-\delta z^2}}$. Thus from \eqref{eq:tauinv}, we have
\begin{align}\label{eq:tauinv1}
(X^{-1})_{q,q'}=&
(-1)^{q-q'} \frac{24 \delta v}{\sqrt{4-\delta z^2}}\left(\sqrt{\frac{2+\delta z}{2-\delta z }} \right)^{q-q'}\frac{p_{(q\wedge q')-1} \times p_{N-1-(q\vee q')} }{p_{N-1}}\left(\frac{10}{\sqrt{4-\delta z^2}} \right).
\end{align}
Since $\delta z<2$, $y_1$ and $y_3$ have an identical sign which is opposite of $y_2$. Moreover from (\ref{eq:tauinv}), $(X^{-1})_{q,q'-1}$ and $(X^{-1})_{q,q'+1}$ are having identical sign which is opposite of $(X^{-1})_{q,q'}$. Hence $(X^{-1})_{q,q'-1}y_1$, $(X^{-1})_{q,q'}y_2$, and $(X^{-1})_{q,q'+1}y_3$ all have an identical sign. Therefore the absolute value of their sum is equal to the sum of their absolute values. Thus for $1\leq q,q' \leq N-1$, we write 
\begin{align}\label{eq:wij2}
\abs{W_{q,q'}}=\abs{(X^{-1})_{q,q'-1}y_3}+\abs{(X^{-1})_{q,q'}y_2}+\abs{(X^{-1})_{q,q'+1}y_1}.
\end{align}
Expressions \eqref{eq:tauinv1}, and \eqref{eq:wij2} are useful in finding the Gershgorin disks $D(a_q,r_q)$ corresponding to the $q^{th}$ row, for locating the eigenvalues of $W$. Indeed, the center $a_q$ and radius $r_q$ for all $1\leq q \leq N-1$ are given by
\begin{align}
a_q&=W_{q,q}=(X^{-1})_{q,q-1}y_3+(X^{-1})_{q,q}y_2+(X^{-1})_{q,q+1}y_1,\label{eq:aq}\\
r_q&=\sum_{q'\neq q}\abs{W_{q,q'}}\label{eq:rq}.
\end{align}
In order to show positivity of the real part of eigenvalues, it is sufficient to show that $a_q > r_q$ for all $1\leq q \leq 2$. To start with $q=1$, we get
\begin{align}
a_1&= (X^{-1})_{1,1}y_2+(X^{-1})_{1,2}y_1 \nonumber\\
&= \frac{24 \delta v}{\sqrt{4-\delta z^2}} \left[\frac{p_{0}p_{1}}{p_{2}}\left(\frac{10}{\sqrt{4-\delta z^2}}\right)y_2-\left(\sqrt{\frac{2+\delta z}{2-\delta z }} \right)^{-1}\frac{p_{0}p_{0}}{p_{2}}\left(\frac{10}{\sqrt{4-\delta z^2}}\right)y_1\right]\nonumber\\
&= \frac{24 \delta v}{\sqrt{4-\delta z^2}} \left[ \frac{y_2 p_{1}} {p_{2}}-\left(\sqrt{\frac{2+\delta z}{2-\delta z }} \right)^{-1} \frac{y_1}{p_{2}}\right] \left(\frac{10}{\sqrt{4-\delta z^2}}\right),\nonumber
\end{align}
using $p_0=1$, and
\begin{align}
\nonumber r_1&=\abs{W_{1,2}}= \abs{(X^{-1})_{1,1}y_3}+\abs{(X^{-1})_{1,2}y_2}\\
&= \frac{24 \delta v}{\sqrt{4-\delta z^2}} \left[ \frac{  p_0p_1}{p_2}\left(\frac{10}{\sqrt{4-\delta z^2}}\right)
\abs{y_3}+ \left(\sqrt{\frac{2+\delta z}{2-\delta z }} \right)^{-1} \frac{  p_0p_0}{p_{2}}\left(\frac{10}{\sqrt{4-\delta z^2}}\right)
\abs{y_2}\right]\nonumber\\
&=\frac{24 \delta v}{\sqrt{4-\delta z^2}}\left[ \frac{-y_{3}p_{1}} {p_{2}} +\left(\sqrt{\frac{2+\delta z}{2-\delta z }} \right)^{-1}\frac{y_{2}} {p_{2}}
\right] \left(\frac{10}{\sqrt{4-\delta z^2}}\right). \nonumber
\end{align}
From above we have $a_1-r_1$ is equal to
\begin{align}\label{a-R}
\left(\frac{24\delta v}{\sqrt{4-\delta z^2}}\right) \left[\frac{y_2 p_1} {p_2}-\left(\sqrt{\frac{2+\delta z}{2-\delta z }} \right)^{-1}\frac{y_1}{p_{2}}+ \frac{y_3p_{1}} {p_{2}} -\left(\sqrt{\frac{2+\delta z}{2-\delta z }} \right)^{-1}\frac{y_2} {p_2}
\right] \left(\frac{10}{\sqrt{4-\delta z^2}}\right). 
\end{align}
Next to prove 
\begin{align} \label{a>R}
    a_1 > r_1
\end{align}
holds for sufficiently small $\delta z$, it is enough to show that $a_1 - r_1\to +\infty$ as $(\delta v, \delta z) \to (0,0)$ by keeping $b$ fixed, because of continuity of \eqref{a-R} w.r.t. $(\delta v, \delta z)$ on $(0,\infty)^2$. By denoting this limit operation as $\displaystyle{\lim_{(\delta v, \delta z) \downarrow (0,0)}}$ and using the expressions \eqref{y123}, and \eqref{a-R}, we have $\displaystyle{\lim_{(\delta v, \delta z) \downarrow (0,0)}} (a_1 - r_1)$ is equal to 
\begin{align*}
\lim_{(\delta v, \delta z) \downarrow (0,0)} \left[(y_2 + y_3)\frac{ p_1} {p_2}-(y_1 + y_2)\frac{1}{p_2}  \right]\left(5\right)
= \lim_{(\delta v, \delta z) \downarrow (0,0)} \left[\left(\frac{c}{2 \delta v} \left(b +\frac{\delta v}{12}\right)\right)\frac{p_1-1} {p_2}  \right]\left(5\right)=\infty,
\end{align*}
as $c>0$, $b>0$, $p_1(5)-1 =9>0$, and $p_2(5)= 99>0$. Thus \eqref{a>R} holds. Similarly, for $q=2$,
\begin{align*}
a_2 &= (X^{-1})_{2,1}y_3+(X^{-1})_{2,2}y_2 \\ 
&= \frac{24 \delta v}{\sqrt{4-\delta z^2}} \left[-\left(\sqrt{\frac{2+\delta z}{2-\delta z }} \right)\left(\frac{p_0 p_0}{p_2}\left(\frac{10}{\sqrt{4-\delta z^2}}\right)\right)y_3+\left(\frac{p_1 p_0}{p_2}\left(\frac{10}{\sqrt{4-\delta z^2}}\right)\right)y_2\right]\\
&= \frac{24 \delta v}{\sqrt{4-\delta z^2}} \left[ \frac{y_2 p_1} {p_2}-\left(\sqrt{\frac{2+\delta z}{2-\delta z }} \right) \frac{y_3}{p_2}\right] \left(\frac{10}{\sqrt{4-\delta z^2}}\right),
\end{align*}
using $p_0=1$, and 
\begin{align*}
\nonumber r_2&=\abs{W_{2,1}}= \abs{(X^{-1})_{2,1}y_2}+\abs{(X^{-1})_{2,2}y_1} \\
&= \frac{24 \delta v}{\sqrt{4-\delta z^2}} \left[\left(\sqrt{\frac{2+\delta z}{2-\delta z }} \right)\frac{  p_0p_0}{p_2}\left(\frac{10}{\sqrt{4-\delta z^2}}\right)
\abs{y_2}+\frac{p_1p_0}{p_{2}}\left(\frac{10}{\sqrt{4-\delta z^2}}\right)
\abs{y_1}\right]\\
&=\frac{24 \delta v}{\sqrt{4-\delta z^2}}\left[ \left(\sqrt{\frac{2+\delta z}{2-\delta z }} \right)\frac{y_2}{p_2} -\frac{y_1 p_1} {p_{2}}
\right] \left(\frac{10}{\sqrt{4-\delta z^2}}\right).
\end{align*}
From above we have, $a_2-r_2$ is equal to
\begin{align}\label{ak-Rk}
\left(\frac{24\delta v}{\sqrt{4-\delta z^2}}\right)\left[ \frac{ y_2 p_1} {p_2}-\left(\sqrt{\frac{2+\delta z}{2-\delta z }} \right) \frac{y_3}{p_{2}}-\left(\sqrt{\frac{2+\delta z}{2-\delta z }} \right) \frac{y_2} {p_{2}}+\frac{y_1p_1} {p_{2}}
\right] \left(\frac{10}{\sqrt{4-\delta z^2}}\right).
\end{align}
Therefore,  $\displaystyle{\lim_{(\delta v, \delta z) \downarrow (0,0)}} (a_2-r_2)$ is equal to 
\begin{align*}
\lim_{(\delta v, \delta z) \downarrow (0,0)} \left[(y_2 + y_1)\frac{ p_1} {p_2}-(y_3 + y_2)\frac{1}{p_2}  \right]\left(5\right)
= \lim_{(\delta v, \delta z) \downarrow (0,0)} \left[\left(\frac{c}{2 \delta v} \left(b +\frac{\delta v}{12}\right)\right)\frac{p_1-1} {p_2}  \right]\left(5\right) = \infty,
\end{align*}
as $c$, $b$, $p_1(5)-1$, and $p_2(5)$ are all positive. Hence, on similar lines of $q=1$ case, we have from above
\begin{align} \label{ak>Rk}
    a_2 > r_2,
\end{align}
for sufficiently small $\delta z$. The result follows from (\ref{a>R}) and (\ref{ak>Rk}).
\end{proof}


\section{Condition Number}\label{sec:condition_number}
Note that the condition number of matrix $(I+W)$ is crucial in precise computation of the matrix $H$. In this section, we obtain an upper bound for the condition number of the matrix $(I+W)$. We first prove an upper bound on $\norm{W}_2$, where $\norm{\cdot}_2$ denotes the spectral norm of a matrix. 
\begin{lemma}\label{lemma:wless1}
If $\delta z <2$, then we have 
\begin{equation}\label{eq:stability1}
\norm{W}_2 < \sqrt{\frac{12 }{5}}\left(\frac{2c}{\delta z^2} +\frac{c}{6} \right)\delta v.
\end{equation}
\end{lemma}
\begin{proof}
As $W=X^{-1}Y$, we obtain upper bounds for $\norm{X^{-1}}_2$ and $\norm{Y}_2$ below. Due to the invertibility of $X^{-1}$, all its singular values are positive. Furthermore, since $(X^{-1})^*X^{-1}=(XX^*)^{-1}$, $s$ is a singular value of $X^{-1}$ iff $s^2$ is an eigenvalue of $(XX^*)^{-1}$, where $X^*$ denotes the transpose of $X$.
Equivalently, $s^{-2}$ is an eigenvalue of the positive definite matrix $Z:=XX^*$, whose entries are as follows
\begin{equation}\label{eq:z}
Z_{q,q'}= \left\{\begin{array}{ll}
\sum_{i=1}^3c_i^2 &\textrm{ if }1 < q'=q < N-1\\
\sum_{i=2}^3c_i^2 &\textrm{ if }q'=q=1 \\
\sum_{i=1}^2c_i^2 &\textrm{ if }q'=q =N-1\\
c_2(c_1+c_3) &\textrm{ if } |q'-q| =1\\
c_1c_3 &\textrm{ if } |q'-q| =2\\
0 &\textrm{ else},
\end{array}\right.
\end{equation}
for all $1 \leq q,q' \leq N-1$. Thus
\begin{equation}
\begin{aligned}
\label{eq:realnormx}
\norm{X^{-1}}_2 = & \max\{s\mid s \textrm{ is a singular value of } X^{-1} \}\\
=& \left(\max\{s^2\mid s^{-2} \textrm{ is an eigenvalue of }Z \}\right)^{\frac 1 2}\\
=&\frac{1}{\left(\min\{s^{-2}\mid s^{-2} \textrm{ is an eigenvalue of }Z \}\right)^{\frac 1 2}}\\
=& \frac{1}{\sqrt{\rho_{min}(Z)}},
\end{aligned}
\end{equation}
where $\rho_{min}(Z)$ is the minimum of the eigenvalue of $Z$. A lower bound of spectrum of $Z$ will be obtained using Gerschgorin’s Circle Theorem (GCT). To facilitate the application of the theorem, the centers and radius of Gerschgorin's disks are calculated below. From \eqref{eq:z}, it is clear that we need to consider four different discs, namely $$D\left(\sum_{i=2}^3c_i^2, c_2(c_1+c_3)+c_1c_3\right), \quad \quad D\left(\sum_{i=1}^3c_i^2, 2c_2(c_1+c_3)+c_1c_3\right),$$ $$D\left(\sum_{i=1}^3c_i^2, 2c_2(c_1+c_3)+2c_1c_3\right),\textrm{ and } D\left(\sum_{i=1}^2c_i^2,c_2(c_1+c_3)+c_1c_3\right).$$ 
The values of $(a-r)$ for all four discs are $$\frac{1}{\delta v^2}\left(\frac{5}{9}-\frac{(2-\delta z)\delta z}{288}\right), \quad \frac{1}{\delta v^2}\left(\frac{4}{9}-\frac{(4-\delta z^2)}{192}\right),$$ $$\frac{1}{\delta v^2}\left(\frac{4}{9}-\frac{(4-\delta z^2)}{144}\right), \textrm{ and } \frac{1}{\delta v^2}\left(\frac{5}{9}+\frac{(2+\delta z)\delta z}{288}\right),$$ respectively. Since the third member is the least, by applying GCT we get 
\begin{align*}
\rho_{min}(Z) \geq \frac{1}{\delta v^2}\left(\frac{4}{9}-\frac{(4-\delta z^2)}{144}\right) =\frac{60 +\delta z^2}{(12\delta v)^2} > \frac{5}{12 (\delta v)^2}.
\end{align*}
Therefore, 
\begin{equation}
\label{eq:normX}
\norm{X^{-1}}_2 < \sqrt{\frac{12}{5}}\delta v.
\end{equation}
Further, using \eqref{y123} we have
\begin{equation}
\begin{aligned}\label{eq:yinnorm}
    \norm{Y}_{\infty}&=\max_{1 \leq q\leq N-1}\sum_{q'=1}^{N-1}\abs{Y_{q,q'}}\\&=\abs{y_1}+\abs{y_2}+\abs{y_3}\\
    & =\frac{1}{\delta v}\left[ \frac{bc}{2}+\frac{c\delta v}{24}+bc+\frac{c\delta v}{12}+\frac{bc}{2}+\frac{c\delta v}{24}\right]\\&=\frac{2bc}{\delta v} +\frac{c}{6},
\end{aligned}
\end{equation}
since $\delta z <2$. Similarly, we have
\begin{equation}
\label{eq:y1norm}
\norm{Y}_1=\max_{1 \leq q'\leq N-1}\sum_{q=1}^{N-1}\abs{Y_{q,q'}} =\abs{y_1}+\abs{y_2}+\abs{y_3} = \frac{2bc}{\delta v} +\frac{c}{6}.
\end{equation}
Using the relation between matrix norms, 
\begin{equation}
\label{eq:normY}
\norm{Y}_2 \leq \sqrt{\norm{Y}_1\norm{Y}_{\infty}},
\end{equation}
we have
$$\norm{Y}_2 \leq \left(\frac{2bc}{\delta v} +\frac{c}{6} \right).$$
\noindent Given $W=X^{-1}Y$, and $b=\frac{\delta v}{\delta z^2}$, we have
\begin{align*}
\norm{W}_2 \leq \norm{X^{-1}}_2\norm{Y}_2 < \sqrt{\frac{12 }{5}}\left(\frac{2c}{\delta z^2} +\frac{c}{6} \right)\delta v.
\end{align*}
\end{proof}

\begin{table}
\caption{The expressions for $a_q$ and $r_q$ from (\ref{eq:aq})-(\ref{eq:rq}) for $4\le q\le 6$, and verification of $a_q-b_q>0$, $4\le q\le 6$, with parameters $\alpha_1=0.25$, $\alpha_2=0.1$, $x_l=0$, $x_r=1$, $\delta v=0.1$, and $\delta z=1/N$.}
  \centering
    \begin{tabular}{|c|p{12cm}|c|}
        \hline
        & \centering{$a_q$ and $r_q$} & $a_q-r_q$\\
        \hline \hline
        &&\\
        $N=4$& $a_1=(X^{-1})_{1,1}y_2+(X^{-1})_{1,2}y_1,$ & 2.02\\
        & $r_1=|(X^{-1})_{1,1}y_3|+|(X^{-1})_{1,2}|(|y_2|+|y_3|)+|(X^{-1})_{1,3}|(|y_1|+|y_2|)$&\\
        &&\\
        & $a_2=(X^{-1})_{2,1}y_3+(X^{-1})_{2,2}y_2+(X^{-1})_{2,3}y_1,$ &2.58 \\
        & $r_2=|(X^{-1})_{2,1}y_2|+|(X^{-1})_{2,2}|(|y_1|+|y_3|)+|(X^{-1})_{2,3}y_2|$ &\\
        &&\\
        & $a_3=(X^{-1})_{3,2}y_3+(X^{-1})_{3,3}y_2,$ &1.87\\
        & $r_3=|(X^{-1})_{3,1}|(|y_2|+|y_3|)+|(X^{-1})_{3,2}|(|y_1|+|y_2|)+|(X^{-1})_{3,3}y_1|$ &\\
        \hline
        &&\\
        $N=5$&  $a_1=(X^{-1})_{1,1}y_2+(X^{-1})_{1,2}y_1,$ & 3.14\\
        &$r_1=|(X^{-1})_{1,1}y_3|+|(X^{-1})_{1,2}|(|y_2|+|y_3|)+|(X^{-1})_{1,3}|(|y_1|+|y_2|+|y_3|)+|(X^{-1})_{1,4}|(|y_1|+|y_2|)
        $ &\\
        &&\\
        & $a_2=(X^{-1})_{2,1}y_3+(X^{-1})_{2,2}y_2+(X^{-1})_{2,3}y_1,$ &4.16\\ 
        & $r_2=|(X^{-1})_{2,1}y_2|+|(X^{-1})_{2,2}|(|y_1|+|y_3|)+|(X^{-1})_{2,3}|(|y_2|+|y_3|)+|(X^{-1})_{2,4}|(|y_1|+|y_2|)$ &\\
        &&\\
        & $a_3=(X^{-1})_{3,2}y_3+(X^{-1})_{3,3}y_2+(X^{-1})_{3,4}y_1,$ &4.11\\
        &  $r_3=|(X^{-1})_{3,1}|(|y_2|+|y_3|)+|(X^{-1})_{3,2}|(|y_1|+|y_2|)+|(X^{-1})_{3,3}|(|y_1|+|y_3|)+|(X^{-1})_{3,4}y_2|$& \\
        &&\\
        & $a_4=(X^{-1})_{4,3}y_3+(X^{-1})_{4,4}y_2$ &2.95\\
         &$r_4=|(X^{-1})_{4,1}|(|y_2|+|y_3|)+|(X^{-1})_{4,2}|(|y_1|+|y_2|+|y_3|)+|(X^{-1})_{4,3}|(|y_1|+|y_2|)+|(X^{-1})_{4,4}y_1|$ &\\
        \hline
          &&\\
        $N=6$&  $a_1=(X^{-1})_{1,1}y_2+(X^{-1})_{1,2}y_1,$ & 4.49\\
        &$r_1=|(X^{-1})_{1,1}y_3|+|(X^{-1})_{1,2}|(|y_2|+|y_3|)+|(X^{-1})_{1,3}|(|y_1|+|y_2|+|y_3|)+|(X^{-1})_{1,4}|(|y_1|+|y_2|+|y_3|)
        +|(X^{-1})_{1,5} |(|y_1|+|y_2|)$ &\\
        &&\\
        & $a_2=(X^{-1})_{2,1}y_3+(X^{-1})_{2,2}y_2+(X^{-1})_{2,3}y_1,$ &5.96\\ 
        & $r_2=|(X^{-1})_{2,1}y_2|+|(X^{-1})_{2,2}|(|y_1|+|y_3|)+|(X^{-1})_{2,3}|(|y_2|+|y_3|)+|(X^{-1})_{2,4}|(|y_1|+|y_2|+|y_3|)+|(X^{-1})_{2,5}|(|y_1|+|y_2|)$ &\\
        &&\\
        & $a_3=(X^{-1})_{3,2}y_3+(X^{-1})_{3,3}y_2+(X^{-1})_{3,4}y_1,$ &6.05\\
        &  $r_3=|(X^{-1})_{3,1}|(|y_2|+|y_3|)+|(X^{-1})_{3,2}|(|y_1|+|y_2|)+|(X^{-1})_{3,3}|(|y_1|+|y_3|)+|(X^{-1})_{3,4}|(|y_2|+|y_3|)+|X^{-1}_{3,5}|(|y_1|+|y_2|)$& \\
        &&\\
        & $a_4=(X^{-1})_{4,3}y_3+(X^{-1})_{4,4}y_2+(X^{-1})_{4,5}y_1$ &5.91\\
         &$r_4=|(X^{-1})_{4,1}|(|y_2|+|y_3|)+|(X^{-1})_{4,2}|(|y_1|+|y_2|+|y_3|)+|(X^{-1})_{4,3}|(|y_1|+|y_2|)+|(X^{-1})_{4,4}|(|y_1|+|y_3|)+|(X^{-1})_{4,5}y_2|$ &\\
        &&\\
        & $a_5=(X^{-1})_{5,4}y_3+(X^{-1})_{5,5}y_2$ &4.26\\
         &$r_5=|(X^{-1})_{5,1}|(|y_2|+|y_3|)+|(X^{-1})_{5,2}|(|y_1|+|y_2|+|y_3|)+|(X^{-1})_{5,3}|(|y_1|+|y_2|+|y_3|)+|(X^{-1})_{5,4}|(|y_1|+|y_2|)+|(X^{-1})_{5,5}y_1|$ &\\
        \hline
    \end{tabular}
    \label{tab:aq_bq}
\end{table}

\begin{theorem}\label{theorem:cond}
If real part of eigenvalues of $W$ are positive, then the condition number of the matrix $(I+W)$ is $\mathcal{O}\left(\frac{\delta v}{\delta z^2}\right)$.
\end{theorem}

\begin{proof}
Note that if $\lambda$ is an eigenvalue of $W$, then $\frac{1}{1+\lambda}$ is an eigenvalue of $(I+W)^{-1}$. Given that real part of $\lambda$ is positive, we have $\abs{\frac{1}{1+\lambda}}<1$. It gives $\norm{(I+W)^{-1}}_2<1$. Thus condition number of matrix $(I+W)=\norm{(I+W)}_2\norm{(I+W)^{-1}}_2$, which is upper bounded by $1+\norm{W}_2$. Using Lemma \ref{lemma:wless1}, an upper bound on condition number of $(I+W)$ is $1+\sqrt{\frac{12 }{5}}\left(\frac{2c}{\delta z^2} +\frac{c}{6} \right)\delta v$. 
\end{proof}

\begin{table}
		\caption{The minimum values of real part of eigenvalues of $W$ for the parameters $\alpha_1=0.25$, $\alpha_2=0.1$, $T=2$, $x_l=0$, $x_r=1$, and various values of $\delta z$ and $\delta v$.}
	\begin{tabular}{|c||P{1.4cm}|P{1.5cm}|P{1.5cm}|P{1.5cm}|P{1.5cm}|P{1.5cm}|P{1.5cm}|P{1.5cm}|P{1.8cm}}
    \hline
        \backslashbox{$\delta z$}{$\delta v$} & $10^{-8}$ & $10^{-7}$ & $10^{-6}$ & $10^{-5}$ & $10^{-4}$ & $10^{-3}$ & $10^{-2}$ & $10^{-1}$ \\
		\hline\hline
		$1/8$ & 3.16e-08 & 3.16e-07 & 3.16e-06 & 3.16e-05 & 3.16e-04 & 3.16e-03 & 3.16e-02 & 3.16e-01 \\
		\hline
		$1/16$ & 3.16e-08 & 3.16e-07 & 3.16e-06 & 3.16e-05 & 3.16e-04 & 3.16e-03 & 3.16e-02 & 3.16e-01 \\
		\hline
		1/32 & 3.16e-08 & 3.16e-07 & 3.16e-06 & 3.16e-05 & 3.16e-04 & 3.16e-03 & 3.16e-02 & 3.16e-01 \\
		\hline
		1/64 & 3.16e-08 & 3.16e-07 & 3.16e-06 & 3.16e-05 & 3.16e-04 & 3.16e-03 & 3.16e-02 & 3.16e-01 \\
		\hline 
		1/128 & 3.16e-08 & 3.16e-07 & 3.16e-06 & 3.16e-05 & 3.16e-04 & 3.16e-03 & 3.16e-02 & 3.16e-01 \\
		\hline
		1/256 & 3.16e-08 & 3.16e-07 & 3.16e-06 & 3.16e-05 & 3.16e-04 & 3.16e-03 & 3.16e-02 & 3.16e-01 \\
		\hline 
		1/512 & 3.16e-08 & 3.16e-07 & 3.16e-06 & 3.16e-05 & 3.16e-04 & 3.16e-03 & 3.16e-02 & 3.16e-01 \\
		\hline
		1/1024 &3.16e-08 & 3.16e-07 & 3.16e-06 & 3.16e-05 & 3.16e-04 & 3.16e-03 & 3.16e-02 & 3.16e-01 \\
		\hline
		1/2048 & 3.16e-08 & 3.16e-07 & 3.16e-06 & 3.16e-05 & 3.16e-04 & 3.16e-03 & 3.16e-02 & 3.16e-01 \\
		\hline
		1/4096 & 3.16e-08 & 3.16e-07 & 3.16e-06 & 3.16e-05 & 3.16e-04 & 3.16e-03 & 3.16e-02 & 3.16e-01 \\
		\hline
	 	\end{tabular}
	\label{table:eigen}
\end{table}

\section{Numerical Illustrations}\label{sec:numerical}
\par This section presents numerical experiments to illustrate the results of the preceding section and to verify the assumptions numerically. We proved theoretically in Remark \ref{remark:N2}, and Proposition \ref{lemma:eigenless1} that real part of eigenvalues of $W$ are positive for $N=2$, and $3$ respectively. Further, the expressions of $a_q$'s, and $b_q$'s are derived in Table \ref{tab:aq_bq} for $N=4,5,$ and $6,$ in terms of the entries of matrices $X^{-1}$ and $Y$. It is observed that for each values of $N,$ $(a_q-b_q)$'s are positive, i.e. the real part of the eigenvalues are positive. Hence, the Crank-Nicolson compact scheme (\ref{eq:xy}) for solving convection-diffusion equation is stable for the parameters considered in Table \ref{tab:aq_bq}. It is observed from the third column of Table \ref{tab:aq_bq} that $a_q-r_q$ is lowest for the last disk for $4\le N\le 6$. Further study is needed to verify if this assertion is true for general $N$. If so, the estimation of $a_q-r_q$ for the last disk will be sufficient to conclude the stability result. Additionally, for higher values of $N$, the positivity of real parts of eigenvalues of $W$ is verified numerically. To this end, Table \ref{table:eigen} presents the minimum values of the real parts of eigenvalues of $W$, and it is evident that these are positive for wide range of $N$.


\begin{figure}[htbp]
  \centering
 \includegraphics[width=9 cm]{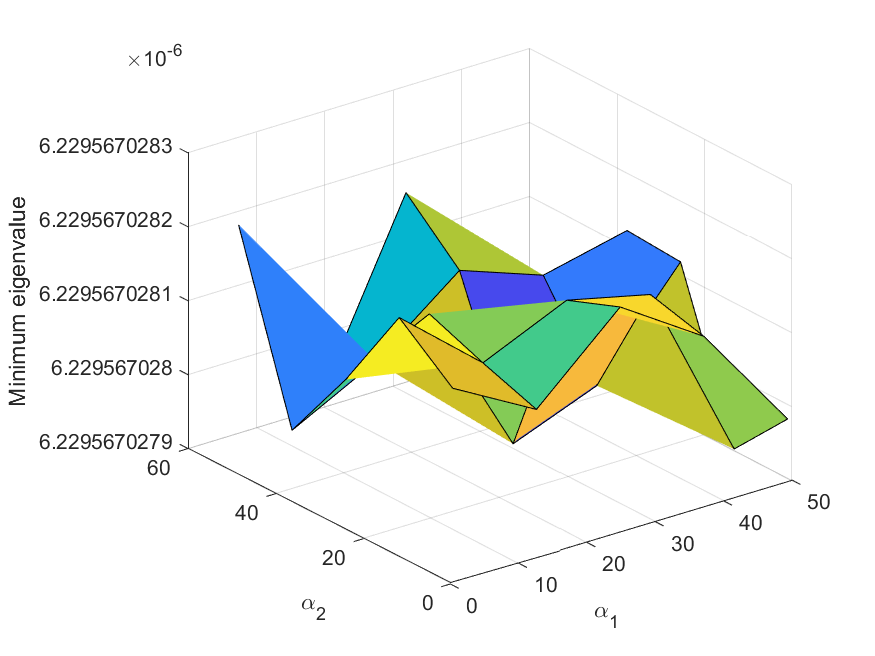}
  \caption{The minimum eigenvalue of $W$ for various values of convection and diffusion coefficients $\alpha_1$ and $\alpha_2$.}
  \label{fig:min}
\end{figure}

\begin{table}[h!]
\caption{The values of an upper bound on $\|X^{-1}\|_2$ using (\ref{eq:normX}) with parameters $\alpha_1=0.25$, $\alpha_2=0.1$, $T=1,$ $x_l=0$, and $x_r=1$. The entries in fourth column are obtained from (\ref{eq:realnormx}).}
    \centering
    \begin{tabular}{|c|c|c|c|}
    \hline
        $N$ & $M$ & Upper bound on &   $\|X^{-1}\|_{2}$ \\
        & & $\|X^{-1}\|_{2},$ using (\ref{eq:normX})&\\
    \hline \hline 
        25 & 800 & $968.25\times10^{-6}$ & $935.66\times 10^{-6}$ \\
        50 & 3200 & $484.12\times 10^{-6}$ & $468.52\times 10^{-6}$\\
        100 & 12800 & $121.03\times 10^{-6}$ & $117.17\times 10^{-6}$ \\
        200 & 51200 & $302.58\times 10^{-7}$ & $292.96\times10^{-7}$\\
        400 & 204800 & $756.44\times 10^{-8}$ & $732.42\times 10^{-8}$ \\
        800 & 819200 & $189.11\times10^{-8}$ & $183.05\times 10^{-8}$ \\
    \hline
    \end{tabular}
    \label{tab:normX}
\end{table}

\begin{table}
    \caption{The values of an upper bound on $\|Y\|_2$ with parameters $\alpha_1=0.25$, $\alpha_2=0.1$, $x_l=0$, and $x_r=1$, where $\|Y\|_1,$ $\|Y\|_2,$ and  $\|Y\|_{\infty}$ are usual notation of matrix norms. The entries in fifth column are $\|Y\|_2=\sqrt{\lambda_{max}(Y^*Y)}$, where $\lambda_{max}$ is the maximum eigenvalue of $Y^{*}Y$.}
    \centering
    \begin{tabular}{|c|c|c|c|c|}
    \hline
        $N$ & $\|Y\|_{\infty}$ &   $\|Y\|_{1}$ & Upper bound on  &  $\|Y\|_2$\\
        & using (\ref{eq:yinnorm}) & using (\ref{eq:y1norm}) & $\|Y\|_2$ using (\ref{eq:normY})& \\
    \hline \hline 
        25 & 781.35 & 781.35 & 781.35 & 778.27\\
        50 & 3125.10 & 3125.10 & 3125.10 & 3122.02\\
        100 & 12500.10  & 12500.10 & 12500.10 & 12497.02\\
        200 & 50000.10 &  50000.10 &  50000.10 & 49997.02\\
        400 & 200000.10 & 200000.10 & 200000.10 & 199997.02\\
        800 & 800000.10 & 800000.10 & 800000.10 & 799997.02\\
    \hline
    \end{tabular}
    \label{tab:normY}
\end{table}


\par It is obvious to observe that the convection and diffusion coefficients $\alpha_1$ $\&$ $\alpha_2$ have influence over the eigenvalues of $W$. To see this, the minimum values of real parts of eigenvalues of $W$ are plotted in Fig. \ref{fig:min} for various values of $\alpha_1$, $\alpha_2$. Note that $\delta v=\sqrt{\frac{5 }{12}} \Big/ \left(\frac{2c}{\delta z^2} +\frac{c}{6} \right)$ is chosen according to Lemma \ref{lemma:wless1}, and $\delta z$ is fixed as $\frac{1}{512}$ for this computation. It is clear from Fig. \ref{fig:min} that the real parts of eigenvalues of $W$ are positive for various values of $\alpha_1$ and $\alpha_2$.

\par Given the matrix $X$ in (\ref{eq:matrix}), a theoretical upper bound on $\|X^{-1}\|_2$ is derived in (\ref{eq:normX}) in terms of $\delta v$. Using this, the numerical values of an upper bound on $\|X^{-1}\|_2$ for the parameters $\alpha_1=0.25$, $\alpha_2=0.1$, $T=1,$ $x_l=0$, and $x_r=1$ are computed. These values are listed in third column of Table \ref{tab:normX} for various values of $M$ and $N$. The entries in the fourth column of the Table \ref{tab:normX} are obtained from (\ref{eq:realnormx}). It is evident that upper bound is reasonably sharp for the given set of parameters. This bound helps us to find an expression of the upper bound on the condition number of $(I+W)$ in terms of discretization parameters.

\par In a similar way, the values of an upper bound on $\|Y\|_2$ are presented in Table \ref{tab:normY} for parameters $\alpha_1=0.25$, $\alpha_2=0.1$, $x_l=0$, and $x_r=1$. Given matrix $Y$ in (\ref{eq:matrix}), the entries of second and third column in Table \ref{tab:normY} are computed using (\ref{eq:yinnorm}) and (\ref{eq:y1norm}), respectively for various values of $N$. The upper bound of $\|Y\|_2$ is obtained using (\ref{eq:normY}) and listed in fourth column of Table \ref{tab:normY}. The entries in fifth column are obtained from the expression $\|Y\|_2=\sqrt{\lambda_{max}(Y^*Y)}$, where $\lambda_{max}$ is the maximum eigenvalue of $Y^{*}Y$. The upper bound on $\|Y\|_2$ also plays an important role in studying the condition number of $(I+W).$

\par Note that the condition number of $(I+W)$, i.e. $\kappa(I+W)$, is crucial in solving (\ref{eq:w}). Utilizing the upper bounds on $\|X^{-1}\|_2$ and $\|Y\|_2,$ an upper bound on the condition number of $(I+W)$ is derived in Lemma \ref{lemma:wless1}. The upper bound on $\kappa(I+W)$ is obtained for parameters $\alpha_1=0.25$, $\alpha_2=0.1$, $T=1,$ $x_l=0$, and $x_r=1$. The values are listed in Table \ref{tab:contion_num} for various choices of $N$ and $M$ so that $\frac{\delta v}{\delta z^2}$ is constant. It is observed that the upper bound on the condition number is reasonably small, which asserts the robustness of the proposed numerical scheme. Moreover, in view of the last column of Table \ref{tab:contion_num}, the derived upper bound is reasonably sharp also.

\begin{table}[h!]
     \caption{Upper bound for $\kappa(I+W)$, i.e. condition number of $(I+W)$, with parameters $\alpha_1=0.25$, $\alpha_2=0.1$, $T=1,$ $x_l=0$, and $x_r=1$. The condition numbers given in the last column are computed using MATLAB.}
    \centering
    \begin{tabular}{|c|c|c|c|c|c|}
    \hline
         $N$ & $M$ & Upper bound & Upper bound & Upper bound for & $\kappa(I+W)$ \\
        & & of $\|X^{-1}\|_{2}$ & of $\|Y\|_2$  & $\kappa(I+W)$&  \\
        & & & & from Theorem \ref{theorem:cond}&\\
    \hline \hline 
        25 & 800 & $968.25\times10^{-6}$ & 781.35 & 1.76 & 1.72\\
        50 & 3200 & $484.12\times 10^{-6}$ & 3125.10  & 2.51 & 2.46\\
        100 & 12800 & $121.03\times 10^{-6}$ & 12500.10 & 2.51 & 2.46\\
        200 & 51200 & $302.58\times 10^{-7}$ & 50000.10 & 2.51 & 2.46\\
        400 & 204800 & $756.44\times 10^{-8}$ & 200000.10 & 2.51 & 2.46\\
        800 & 819200 & $189.11\times10^{-8}$ & 800000.10 & 2.51 & 2.46\\
    \hline
    \end{tabular}
    \label{tab:contion_num}
\end{table}

\begin{table}[h!]
	\caption{Numerical verification of validity and sharpness of the inequality $\frac{\norm{W}_{2}}{\left(\sqrt{\frac{12 }{5}}\left(\frac{2c}{\delta z^2} +\frac{c}{6} \right)\delta v\right)} < 1$ in Lemma \ref{lemma:wless1} for various values of $\delta z$ and $\delta v$.}
	\begin{tabular}{|c||P{1cm}|P{1.5cm}|P{1.5cm}|P{1.5cm}|P{1.5cm}|P{1.5cm}|P{1.5cm}|P{1.5cm}|P{1.8cm}}
    \hline
        \backslashbox{$\delta z$}{$\delta v$} & $10^{-8}$ & $10^{-7}$ & $10^{-6}$ & $10^{-5}$ & $10^{-4}$ & $10^{-3}$ & $10^{-2}$ & $10^{-1}$ \\
		\hline\hline
		$1/8$ & 0.9140 & 0.9140 & 0.9140 & 0.9140 & 0.9140 & 0.9140 & 0.9140 & 0.9140 \\
		\hline
		$1/16$ & 0.9543 & 0.9543 & 0.9543 & 0.9543 & 0.9543 & 0.9543 & 0.9543 & 0.9543 \\
		\hline
		1/32 & 0.9647 & 0.9647 & 0.9647 & 0.9647 & 0.9647 & 0.9647 & 0.9647 & 0.9647 \\
		\hline
		1/64 & 0.9673 & 0.9673 & 0.9673 & 0.9673 & 0.9673 & 0.9673 & 0.9673 & 0.9673 \\
		\hline 
		1/128 & 0.9680 & 0.9680 & 0.9680 & 0.9680 & 0.9680 & 0.9680 & 0.9680 & 0.9680 \\
		\hline
		1/256 & 0.9681 & 0.9681 & 0.9681 & 0.9681 & 0.9681 & 0.9681 & 0.9681 & 0.9681 \\
		\hline 
		1/512 & 0.9682 & 0.9682 & 0.9682 & 0.9682 & 0.9682 & 0.9682 & 0.9682 & 0.9682 \\
		\hline
		1/1024 & 0.9682 & 0.9682 & 0.9682 & 0.9682 & 0.9682 & 0.9682 & 0.9682 & 0.9682 \\
		\hline
		1/2048 & 0.9682 & 0.9682 & 0.9682 & 0.9682 & 0.9682 & 0.9682 & 0.9682 & 0.9682 \\
		\hline
		1/4096 & 0.9682 & 0.9682 & 0.9682 & 0.9682 & 0.9682 & 0.9682 & 0.9682 & 0.9682 \\
		\hline
	 	\end{tabular}
	\label{table:norm}
\end{table}

\par Note that in Lemma \ref{lemma:wless1}, the expression of an upper bound of $\norm{W}_2$ is derived in terms of $\delta z$, $\delta v,$ and $c$. To investigate the sharpness of this bound, another numerical experiment is performed. Table \ref{table:norm} reports the ratio of $\norm{W}_2$ and upper bound for the various values of $\delta z$ and $\delta v$. The computed ratios are found ranging from $91\%$ to $97\%$ corresponding to the chosen values of $\delta z$ and $\delta v$. This indicates that the proposed upper bound on $\norm{W}_2$ is significantly sharp, which ensures the sharpness of the upper bound on $\kappa(I+W)$ derived in Theorem \ref{theorem:cond}.

\section{Conclusions and future directions}\label{sec:conclu}
A matrix method approach has been developed to establish the stability of Crank-Nicolson compact scheme for convection-diffusion equations under certain assumptions. The application of Gerschgorin Circle Theorem and the expression of inverse of tridiagonal Toeplitz matrix has facilitated in verifying those assumptions. Illustrations have been provided to validate those assumptions for larger range of discretization parameters. Since computation of amplification matrix requires the inversion of a matrix, an upper bound on the condition number of that matrix is derived. Numerical examples have been considered to study the sharpness of that proposed upper bound. As a future work, the proposed methodology may be investigated for variable coefficient problems, multi-dimensional problems, and for system of PDEs etc.\\

\noindent {\bf Conflict of interest:} The authors declare that they have no conflict of interest.

\bibliographystyle{elsarticle-num}
\bibliography{references}
\end{document}